\documentclass[12pt]{article}

\usepackage{amsmath,amsfonts,amssymb,amsthm}
\usepackage[isolatin]{inputenc}

\newtheorem{theorem}{Theorem}

\theoremstyle{remark}
\newtheorem{remark}{Remark}

\theoremstyle{definition}
\newtheorem{definition}{Definition}

\begin{document}

\title{Fractional Action-Like Variational Problems\footnote{This
is a preprint of an article accepted for publication (28-04-2008)
in \emph{J. Math. Phys.},  AIP, \texttt{http://jmp.aip.org/jmp}}}

\author{Rami Ahmad El-Nabulsi\\[0.1cm]
        Plasma Application Laboratory\\
        Department of Nuclear and Energy Engineering\\
        Faculty of Mechanical, Energy and Production Engineering\\
        Cheju National University\\
        Ara-dong 1, Jeju 690-756, South Korea\\
        \texttt{nabulsiahmadrami@yahoo.fr}\\[0.3cm]
\and
        Delfim F. M. Torres\\[0.1cm]
        Centre for Research on Optimization and Control\\
        Department of Mathematics, University of Aveiro\\
        Campus Universitário de Santiago\\
        3810-193 Aveiro, Portugal\\
        \texttt{delfim@ua.pt}}

\date{}

\maketitle


\begin{abstract}
Fractional action-like variational problems have recently gained
importance in studying dynamics of nonconservative systems. In
this note we address multi-dimensional
fractional action-like problems of the calculus of variations.
\end{abstract}

\medskip


\smallskip



\section{Introduction}

Fractional calculus (FC) represents a powerful tool in applied
mathematics to study a myriad of problems from different fields of
science and engineering, with many break-through results found in
mathematical physics, finance, hydrology, biophysics,
thermodynamics, control theory, statistical mechanics,
astrophysics, cosmology and bioengineering
\cite{8,7,9,5,6,6b,1,2,3,4}. Although various fields of
application of fractional derivatives and integrals are already
well established, some others have just started. An example of the
later is the study of fractional problems of the calculus of
variations and respective Euler-Lagrange type equations
\cite{10,16,17,cre,24}.

During last years, fractional variational principles (FVP) were
proposed and several applications given. Different methods were
used to obtain the fractional Euler-Lagrange equations and the
corresponding Hamiltonian canonical equations
\cite{10,15,16,17,cre,24,11,12,13,18,14}. The existence of several
different FVP and the need for still more elaboration on the
subject, both for classical and quantized systems, is partially
explained by the nonlocal nature of the fractional differential
operators used to describe the dynamics, and the form of the
corresponding adjoint operators. Another reason is the existence
of many different fractional integral operators, including the
ones of Grünwald-Letnikov, Caputo, Riesz, and
Riemann-Liouville. The Riemann-Liouville operator is one of the
most frequently used when fractional integration is performed
\cite{9,2}.

Recently, the first author has proposed a new one-dimensional (1D)
approach, known as the fractional action-like variational approach
(FALVA), in order to model nonconservative dynamical systems
\cite{20,21}. In FALVA the fractional time integral introduces
(only) one real parameter $\alpha$, and the derived Euler-Lagrange
equations are simpler and similar to the standard ones. The
novelty in the Euler-Lagrange equations is the presence of a
fractional generalized external force acting on the system
\cite{20,21,22}. In particular, no fractional derivatives appear
in the derived equations. The conjugate momentum, the Hamiltonian
and Hamilton's equations are shown to depend on the fractional
order of integration $\alpha$. Constants of motion for fractional
action-like variational problems were discussed in \cite{23,23a}
(see also \cite{23b,23c}); FALVA problems with higher-order
derivatives are studied in \cite{23a,23d}; and some encouraging
results obtained and discussed in \cite{27,29,31,33,24}. Here we
are interested to generalize FALVA for multiple fractional
action-like integrals of the calculus of variations
(multi-dimensional fractional action-like variational problems).

The text is organized as follows: in Sect.~\ref{sec:2} we review
the basic concepts of 1D-FALVA; the new extensions are found in
Sect.~\ref{sec:mr}. Since various applications of FC are based on
replacing the time derivative in dynamical equations with a
fractional derivative, in \S\ref{sec:3} we introduce the
double-weighted FVP with the recent fractional derivatives of
J.~Cresson \cite{cre}. In \S\ref{sec:4} the FVP is given in
general form, for the arbitrary $N$-dimensional case.
Sect.~\ref{sec:5} is dedicated to conclusions and future
perspectives.


\subsection*{A Note on the notation used}

In order to be clear when $f(t)$ stands for a function or the value of the function at a point $t$, we denote the function by $t \to f(t)$ and the value of the function at $t$ by $f(t)$. In the notation $t \to f(t)$, $t$ is a dummy variable. For instance, in the beginning of \S\ref{subsec:pfd} we write $(\dot{q},q,\tau) \to L(\dot{q},q,\tau)$, which represents a function of three variables. Exactly the same function can be written, for example, as $(a,b,c) \to L(a,b,c)$ ($\dot{q}$, $q$, $\tau$ or $a$, $b$, $c$ are here dummy variables). However, the dummy variables we choose in the paper are used consistently to fix the notation for the partial derivatives. For example, if we define the function as $(a,b,c) \to L(a,b,c)$, then we write $\frac{\partial L}{\partial a}$ to denote the partial derivative of function $L$ with respect to the first argument; if we define it as $(\dot{q},q,\tau) \to L(\dot{q},q,\tau)$, then the partial derivative of $L$ with respect to the first argument is denoted by $\frac{\partial L}{\partial \dot{q}}$.


\section{Brief Overview of 1D-FALVA}
\label{sec:2}

In this section we summarize the one-dimensional FALVA. The review
consists of two cases: absence of fractional derivatives
(\S\ref{subsec:afd}); and presence of Riemann-Liouville fractional
derivatives in the sense of Cresson \cite{cre}
(\S\ref{subsec:pfd}). The reader is referred respectively to
\cite{20,21} and \cite{24} for more details.


\subsection{Absence of fractional derivatives}
\label{subsec:afd}

Consider a smooth $n$-dimensional manifold $M$ (configuration
space) and denote by $L: TM \times M \times \mathbb{R} \to
\mathbb{R}$ the smooth Lagrangian function. For any smooth path
$q:\left[ {a ,b } \right] \to M$ satisfying fixed boundary
conditions $q(a)=q_a$ and $q(b)=q_b$, we define the fractional
action integral by
\begin{equation}
\label{eq:1} S^{\alpha} \left[ q \right](t) = \frac{1} {{\Gamma
(\alpha )}}\int\limits_{a }^t {L(\dot q(\tau ),q(\tau ),\tau )(t -
\tau )^{\alpha  - 1} d\tau } \, ,
\end{equation}
where $\Gamma$ is the Euler gamma function, $\dot{q} =
\frac{dq}{d\tau}$, $0 < \alpha < 1$, $\tau \in (a,t)$ is the
intrinsic time and $t \in [a,b]$ is the observer time. For a
discussion of the importance to consider a multi-time formalism,
we refer the reader to \cite{CD:UdriDu:2005,CD:UdriTe:2007}. The
Lagrangian $L(\dot q,q,\tau)$ is  weighted by ${{(t - \tau
)^{\alpha - 1} } \mathord{\left/
 {\vphantom {{(t - \tau )^{\alpha  - 1} } {\Gamma (\alpha )}}} \right.
 \kern-\nulldelimiterspace} {\Gamma (\alpha )}}$, and one
can write the smooth action integral $S^{\alpha}\left[ q
\right](t)$ as  $\int\limits_{a}^t {L(\dot q(\tau ),q(\tau ),\tau
)dg_t (\tau )}$, with an appropriate time smeared measure $dg_t
(\tau )$ on the time interval $\left[ {a,t} \right]$, or as the
Riemann-Liouville operator applied to the Lagrangian
$L(q'(t),q(t),t)$ \cite{20,21}. Here we just note that when
$\alpha \to 1$ the functional $S^{\alpha} \left[ q \right](b)$
tends to the classical action integral of the calculus of
variations:
\begin{equation*}
\lim_{\alpha \to 1} S^{\alpha} \left[ q \right](b)=
\int\limits_{a}^b {L(\dot{q}(\tau),q(\tau),\tau)d\tau } \, .
\end{equation*}
Functionals of type \eqref{eq:1} appear naturally in mathematical
economy, where they are used for describing discounting economical
dynamics; and in the theory of dynamical systems, for describing
nonlinear dissipative structures.

The Euler-Lagrange equations associated with the fractional action
integral \eqref{eq:1} take the following form: for all $t \in
[a,b]$,
\begin{equation}
\label{eq:5} \frac{{\partial L}} {{\partial q_i }} - \frac{d}
{{d\tau }}\left( {\frac{{\partial L}} {{\partial \dot{q}_i }}}
\right) = \frac{{1 - \alpha }} {{t - \tau }}\frac{{\partial L}}
{{\partial \dot{q}_i }} \, , \quad i = 1, \ldots, n \, ,
\end{equation}
where the partial derivatives of $L$ are evaluated at
$(\dot{q}(\tau),q(\tau),\tau)$, $\tau \in (a,t)$. The fractional
term $\frac{\alpha - 1}{\tau - t} \frac{{\partial L}} {{\partial
\dot{q}_i }}$ has a meaning in classical, quantum or relativistic
settings, being responsible for adding a time-dependent damping
coefficient into the dynamical equations. If we denote by $R =
{{(1 - \alpha )L} \mathord{\left/
 {\vphantom {{(1 - \alpha )L} {(t - \tau )}}} \right.
 \kern-\nulldelimiterspace} {(t - \tau )}}$
the fractional Rayleigh dissipation function and by $E$ the
Euler-Lagrange operator, \textrm{i.e.} $E =
\frac{\partial}{\partial q} - \frac{d}{d \tau}
\frac{\partial}{\partial \dot{q}}$, then equations \eqref{eq:5}
take the form $E(L) = \frac{\partial R}{\partial \dot{q}}$.
Extremals are defined as solutions of the fractional
Euler-Lagrange equations (FELE) $E(L) = \frac{\partial R}{\partial
\dot{q}}$.


\subsection{Presence of fractional derivatives}
\label{subsec:pfd}

For any smooth path $q:[a ,b] \to M$ satisfying fixed boundary
conditions $q(a) = q_a$ and $q(b) = q_b$, let the fractional
functional associated to $(\dot{q},q,\tau) \to L(\dot{q},q,\tau)$
be now defined by
\begin{equation}
\label{eq:10} S_{\gamma}^{\alpha,\beta}\left[q\right](t) =
\frac{1} {{\Gamma (\alpha )}}\int\limits_{a}^t {L(D_\gamma
^{\alpha ,\beta } q(\tau ),q(\tau ),\tau )(t - \tau )^{\alpha  -
1} d\tau } \, ,
\end{equation}
where $\left[ {a ,t } \right] \subseteq [a,b] \subset \mathbb{R}$,
and $D_\gamma ^{\alpha ,\beta }$ denotes the Riemann-Liouville
derivative of order $\left(\alpha,\beta\right)$ as defined by
J.~Cresson in \cite{cre}: for all $a$, $t \in \mathbb{R}$, $a <
t$, the fractional derivative operator of order $(\alpha ,\beta)$,
$0 < \alpha$, $\beta < 1$, is given by
\begin{equation}
\label{eq:11} D_\gamma ^{\alpha ,\beta }  = \frac{1} {2}\left[
{D_{a_ +  }^\alpha   - D_{t_ -  }^\beta  } \right] +
\frac{{i\gamma }} {2}\left[ {D_{a_ +  }^\alpha   + D_{t_ - }^\beta
} \right] \, ,
\end{equation}
with $\gamma  \in \mathbb{C}$, $i = \sqrt { - 1}$, $D_{a_ +
}^\alpha$ and $D_{t_ - }^\beta$ the usual left and right
Riemann-Liouville fractional derivatives of order $0 < \alpha,
\beta < 1$ \cite{2},
\begin{gather*}
D_{a_ +  }^\alpha  f\left( \theta \right) = \frac{1} {{\Gamma (1 - \alpha )}}{\frac{d} {{d\theta}}} \int\limits_a^\theta {f\left(\tau
\right)\left( {\theta - \tau } \right)^{- \alpha} d\tau } \, , \\
D_{t - }^\beta  f\left( \theta \right) = \frac{1} {{\Gamma (1 - \beta
)}}\left({ - \frac{d} {{d\theta}}}  \int\limits_\theta^t {f\left( \tau
\right)\left( {\tau - \theta} \right)^{- \beta} d\tau }\right) \, ,
\end{gather*}
$\theta \in [a,t]$.
For more details we refer the reader to \cite{cre}. Here we just
note that for $\gamma = i$ one has $D_\gamma ^{\alpha ,\beta } = -
D_{t - }^\beta$; for $\gamma  = - i$, $D_\gamma ^{\alpha ,\beta }
= D_{a + }^\alpha$; for $\alpha \to 1$ and $\beta \to 1$ one
obtains $D_\gamma ^{\alpha ,\beta } = \frac{d}{dt}$ and
\eqref{eq:10} is reduced to the classical functional of the
calculus of variations. In \cite{24} it is proved that the FELE
correspondent to \eqref{eq:10} is given by
\begin{equation}
\label{eq:15} \frac{{\partial L}} {{\partial q}} - D_{ -
\gamma}^{\beta ,\alpha } \left( {\frac{{\partial L}} {{\partial
\dot q}}} \right) = \frac{{1 - \alpha }} {{t - \tau
}}\frac{{\partial L}} {{\partial \dot q}}
\end{equation}
$\forall$ $\tau \in (a,t)$, where the partial derivatives of $L$
are evaluated at $(D_\gamma ^{\alpha ,\beta } q(\tau ),q(\tau
),\tau )$.

In the next section we generalize \eqref{eq:15} for FALVA problems
with multiple integrals.


\section{Main Results}
\label{sec:mr}

More generally, we discuss here the multi-dimensional fractional
variational problem or $N$D-FALVA ($N$-dimensional FALVA). We
first start by double weighted fractional integrals, \textrm{i.e.}
by the 2D-FALVA.

\subsection{Double-Weighted Fractional Variational Principles (2D-FALVA)}
\label{sec:3}

Similar to Sect.~\ref{sec:2}, we denote by $M$ a smooth
$n$-dimensional manifold; the admissible paths are smooth
functions $q:\Omega \subset \mathbb{R}^2 \to M$ satisfying fixed
Dirichlet boundary conditions on $\partial \Omega$; the Lagrangian
function $(q_x,q_y,q,x,y) \to L(q_x,q_y,q,x,y)$ is supposed to be
sufficiently smooth with respect to all its arguments; $\alpha$
and $\beta$ are two real parameters taking values on the interval
$(0,1)$.

\begin{definition}
\label{def:3.1} The 2D-FALVA action integral is defined by
\begin{equation}
\label{eq:16} \frac{1} {\Gamma \left( \alpha  \right)\Gamma
\left(\beta \right)}\iint_{\Omega(\xi,\lambda)} L(D_{\gamma
;x}^{\alpha ,\delta } q,D_{\gamma ;y}^{\beta ,\chi }
q,q,x,y)\left( {\xi - x} \right)^{\alpha  - 1} \left( {\lambda -
y} \right)^{\beta  - 1} dx dy \, ,
\end{equation}
where $\xi$ and $\lambda$ are the observer times, $(\xi,\lambda)
\in \Omega$; $x$ and $y$ are the intrinsic times, $(x,y) \in
\Omega(\xi,\lambda) \subseteq \Omega$, $x \ne \xi$ and $y \ne
\lambda$; $q = q(x,y)$; $D_{\gamma ;x}^{\delta ,\alpha }$ and
$D_{\gamma ;y}^{\chi ,\beta }$ are the fractional derivative
operators \eqref{eq:11}, respectively of order $(\delta ,\alpha)$
with respect to $x$ and of order $(\chi ,\beta )$ with respect to
$y$. We denote \eqref{eq:16} by ${S_{\gamma}^{\alpha ,\beta ,
\delta ,\chi}} \left[ q \right](\xi,\lambda)$.
\end{definition}

\begin{remark}
The classical multi-variable variational calculus has some
limitations which a multi-time variational calculus, like the one
we are proposing here, can successfully overcome. For a discussion
of the importance to consider a multi-time formalism we refer the
reader to the works of C.~Udriste and his collaborators
\cite{CD:UdriDu:2005,CD:UdriTe:2007}.
\end{remark}

The primary objective is to find functions $q = q(x,y)$ that make
the fractional action ${S_{\gamma}^{\alpha ,\beta , \delta ,\chi}}
\left[ q \right](\xi,\lambda)$ stationary for every $(\xi,\lambda)
\in \Omega$.

\begin{theorem}
\label{theo:3.1} Given a smooth Lagrangian $(q_x,q_y,q,x,y)
\rightarrow L(q_x,q_y,q,x,y)$, if $q = q(x,y)$ makes the
fractional action ${S_{\gamma}^{\alpha ,\beta , \delta ,\chi}}
\left[ q \right](\xi,\lambda)$ defined by \eqref{eq:16} stationary
for every $(\xi,\lambda) \in \Omega$, then the following
double-weighted Euler-Lagrange equation holds for every $(x,y) \in
\Omega(\xi,\lambda)$:
\begin{multline}
\label{eq:17} D_{ - \gamma ;x}^{\delta ,\alpha } \left(
{\frac{{\partial L(D_{\gamma ;x}^{\alpha ,\delta } q,D_{\gamma
;y}^{\beta ,\chi } q,q,x,y)}} {{\partial q_x }}} \right) + D_{ -
\gamma ;y}^{\chi ,\beta } \left( {\frac{{\partial L(D_{\gamma
;x}^{\alpha ,\delta } q,D_{\gamma ;y}^{\beta ,\chi } q,q,x,y)}}
{{\partial q_y }}} \right) \\
 + \frac{{1 - \alpha }}
{{\xi - x}}\left( {\frac{{\partial L(D_{\gamma ;x}^{\alpha ,\delta
} q,D_{\gamma ;y}^{\beta ,\chi } q,q,x,y)}} {{\partial q_x }}}
\right) + \frac{{1 - \beta }} {{\lambda - y}}\left(
{\frac{{\partial L(D_{\gamma ;x}^{\alpha ,\delta } q,D_{\gamma
;y}^{\beta ,\chi } q,q,x,y)}}
{{\partial q_y }}} \right) \\
 - \frac{{\partial L(D_{\gamma ;x}^{\alpha,\delta }
 q,D_{\gamma ;y}^{\beta ,\chi } q,q,x,y)}}
{{\partial q}} = 0 \, .
\end{multline}
\end{theorem}

\begin{proof}
Let $q$ be a stationary solution, $h \ll 1$ a small real
parameter, and $w(x,y)$ an admissible variation, \textrm{i.e.} an
arbitrary smooth function satisfying $w(x,y) = 0$ for all $(x,y)
\in \partial \Omega$ so that $q + hw$ satisfies the given
Dirichlet boundary conditions for all $h$. The fractional action
${S_{\gamma}^{\alpha ,\beta ,\delta,\chi}} \left[ q + h w\right]$
can be written as
\begin{multline*}
\frac{1} {{\Gamma \left( \alpha  \right)\Gamma \left( \beta
\right)}}\iint_{\Omega(\xi,\lambda)} L(D_{\gamma ;x}^{\alpha
,\delta } q + hD_{\gamma ;x}^{\alpha ,\delta} w,D_{\gamma
;y}^{\beta ,\chi } q
+ hD_{\gamma ;y}^{\beta ,\chi } w,q + hw,x,y)\\
\times \left( {\xi - x} \right)^{\alpha  - 1} \left( {\lambda - y}
\right)^{\beta  - 1} dx dy \, ,
\end{multline*}
and the stationary condition $\left. {\frac{d}
{{dh}}{S_{\gamma}^{\alpha,\beta,\delta,\chi}}\left[ {q + hw}
\right]} \right|_{h = 0} = 0$ gives
\begin{multline}
\label{eq:ant:pp} \frac{1} {{\Gamma \left( \alpha  \right)\Gamma
\left( \beta \right)}}\iint_{\Omega(\xi,\lambda)}\left(
{w\frac{{\partial L}} {{\partial q}}} \right.  + D_{\gamma
;x}^{\alpha ,\delta } w\frac{{\partial L}} {{\partial q_x }}  +
\left. {D_{\gamma ;y}^{\beta ,\chi } w\frac{{\partial L}}
{{\partial q_y }}}
\right)\\
\times \left( {\xi - x} \right)^{\alpha  - 1} \left(
{\lambda - y} \right)^{\beta  - 1} dx dy = 0\, ,
\end{multline}
where the partial derivatives of function $(q_x,q_y,q,x,y)
\rightarrow L(q_x,q_y,q,x,y)$ are evaluated at $(D_{\gamma
;x}^{\alpha ,\delta } q(x,y) ,D_{\gamma ;y}^{\beta ,\chi }
q(x,y),q(x,y),x,y)$. Using integration by parts and Green's
theorem, we know that
\begin{multline*}
\iint_{\Omega(\xi,\lambda)} \left( {\frac{{\partial P}} {{\partial
\xi}}G_1  + \frac{{\partial P}}
{{\partial \lambda}}G_2 } \right)d\bar \xi d\bar \lambda \\
= \oint_{\partial \Omega } P\left( { - G_2 d\bar \xi + G_1 d\bar
\lambda} \right) - \iint_{\Omega(\xi,\lambda)}\left( {P\left(
{\frac{{\partial G_1 }} {{\partial \xi}} + \frac{{\partial G_2 }}
{{\partial \lambda}}} \right)} \right)d\bar \xi d\bar \lambda
\end{multline*}
for any smooth functions $G_1$ and $G_2$, where
\begin{equation*}
\begin{split}
\Gamma \left( {1 + \alpha } \right)\bar \xi
&= \xi^\alpha   - \left( {\xi - x} \right)^\alpha \, ,\\
\Gamma \left( {1 + \alpha } \right)\bar \lambda &= \lambda^\beta -
\left( {\lambda - y} \right)^\beta \, .
\end{split}
\end{equation*}
We conclude from \eqref{eq:ant:pp} that
\begin{multline*}
- \frac{1} {{\Gamma \left( \alpha  \right)\Gamma \left( \beta
\right)}} \iint_{\Omega(\xi,\lambda)} w\Biggl[ {\left( {\xi - x}
\right)^{\alpha - 1} \left( {\lambda - y} \right)^{\beta  - 1} }
\times \left( {D_{\gamma ;x}^{\alpha ,\delta } \left(
{\frac{{\partial L}} {{\partial q_x }}} \right)} \right.  \left. {
+ D_{\gamma ;y}^{\beta ,\chi } \left( {\frac{{\partial L}}
{{\partial q_y }}} \right)} \right)  \\
+ (1 - \alpha )\left( {\frac{{\partial L}} {{\partial q_x }}}
\right)\left( {\xi - x} \right)^{\alpha - 2} \left( {\lambda - y}
\right)^{\beta  - 1} + (1 - \beta )\left( {\frac{{\partial L}}
{{\partial q_y }}} \right)\left( {\xi - x} \right)^{\alpha - 1}
\left( {\lambda - y} \right)^{\beta  - 2} \\
 { - \frac{{\partial L}}
{{\partial q}}\left( {\xi - x} \right)^{\alpha  - 1} \left(
{\lambda - y} \right)^{\beta  - 1} } \Biggr] dx dy = 0
\end{multline*}
and then, because of the arbitrariness of $w(x,y)$ inside
$\Omega(\xi,\lambda)$, it follows \eqref{eq:17}, which is the
2D-FELE.
\end{proof}

We expect that fractional variational problems involving multiple
integrals may have important consequences in mechanical problems
involving dissipative systems with infinitely many degrees of
freedom.


\subsection{N-Weighted Fractional Variational Principles (ND-FALVA)}
\label{sec:4}

All the arguments of \S\ref{sec:3} can be repeated, \emph{mutatis
mutandis}, to the $N$-dimensional situation when the admissible
paths are smooth functions $q:\Omega  \subset \mathbb{R}^N  \to M$
satisfying given Dirichlet boundary conditions on $\partial
\Omega$.

\begin{definition}
\label{Definition4.3} Consider a smooth manifold $M$ and let
$(q_{x_1},\ldots,q_{x_N},q,x_1,\ldots,x_N) \rightarrow
L(q_{x_1},\ldots,q_{x_N},q,x_1,\ldots,x_N)$ be a sufficiently
smooth Lagrangian function. The $N$D-FALVA functional is defined
by
\begin{equation*}
{S_{\gamma}^{\alpha,\delta}}\left[ q \right](\xi) =
\frac{1} {{\prod\limits_{i = 1}^N {\Gamma \left( {\alpha _i }
\right)}} } \int {\cdots\int_{\Omega(\xi)} } L\left(\nabla_\gamma
^{\alpha,\delta} q(x),q(x),x\right) \prod\limits_{i = 1}^N {\left(
{\xi_i  - x_i } \right)^{\alpha _i - 1} } dx   \, ,
\end{equation*}
where $x= (x_1,\ldots,x_N)$ is the intrinsic time vector, $\xi=
(\xi_1,\ldots,\xi_N) \in \Omega$ the observer time vector, $x \in
\Omega(\xi) \subseteq \Omega$ with $x_i \ne \xi_i$
($i=1,\ldots,N$), $dx = dx_1 \cdots dx_N$, $\alpha =
(\alpha_1,\ldots,\alpha_N)$, $\delta =
(\delta_1,\ldots,\delta_N)$, $0 < \alpha _i< 1$ ($i = 1, \ldots
N$), and $\nabla_\gamma^{\alpha,\delta} = (D_{\gamma;x_1}^{\alpha
_1 ,\delta_1 } ,\ldots,D_{\gamma;x_N}^{\alpha_N ,\delta_N })$.
\end{definition}

\begin{theorem}
\label{Theorem4.2} The $N$-dimensional Euler-Lagrange equation
associated to the fractional functional
${S_{\gamma}^{\alpha,\delta}}\left[ q \right](\xi)$, $\xi \in
\Omega$, is given by
\begin{equation*}
\sum\limits_{i = 1}^N \Biggl[ {D_{ - \gamma ;x_i
}^{\delta _i ,\alpha _i } \left( {\frac{{\partial L}} {{\partial
q_{x_i}}}} \right)} + {\frac{{1 - \alpha _i }} {{\xi_i  - x_i
}}\left( {\frac{{\partial L}} {{\partial q_{x_i}}}} \right)}
\Biggr] - \frac{{\partial L}} {{\partial q}} = 0 \, ,
\end{equation*}
where all partial derivatives of $L$ are evaluated at
$\left(\nabla_\gamma ^{\alpha ,\delta} q(x), q(x), x\right)$, $x
\in \Omega(\xi)$.
\end{theorem}


\section{Conclusions and Further Work}
\label{sec:5}

In this work we introduce the multi-dimensional FALVA setting and
derive the corresponding multi-dimensional fractional
Euler-Lagrange equations. Obtained Euler-Lagrange equations are
enough complicated, and one expects solutions to be found using
numerical techniques. A work in this direction is in progress. We
expect that fractional variational problems involving multiple
integrals will have important consequences in mechanical problems
and optimal control theory involving dissipative systems with
infinitely many degrees of freedom. A geometric formulation of the field equations for the fractional action-like variational
formalism, in terms of multi-vector fields on tangent bundles, is
under investigation.

In our paper, as well as in all previous works on fractional Euler-Lagrange equations we are aware of, it is assumed that at least one stationary point for the fractional functional exist.
Euler-Lagrange equations are valid under this assumption.
The question of obtaining conditions on the Lagrangian $L$
assuring the existence of stationary trajectories is, to the best of our knowledge, a completely open question in the fractional setting. This is a pertinent question because even very simple
Lagrangians, \textrm{e.g.} $L = D_\gamma ^{\alpha ,\beta } q$,
fail to satisfy the hypotheses under which the Euler-Lagrange equations are valid.


\section*{Acknowledgements}

The first author was supported by the Ministry of Commerce,
Industry and Energy, Korea, through the department of Nuclear
Engineering of the Cheju National University; the second author
was supported by the \emph{Portuguese Foundation for Science and
Technology} (FCT) through the \emph{Centre for Research in
Optimization and Control} (CEOC) of the University of Aveiro,
cofinanced by the European Community Fund FEDER/POCI 2010.




{\small

\begin{thebibliography}{15}

\bibitem{10} O. P. Agrawal, Formulation of Euler-Lagrange
equations for fractional variational problems,
J. Math. Anal. Appl. {\bf 272} (2002), no.~1, 368--379.

\bibitem{8} O. P. Agrawal,
Application of fractional derivatives
in thermal analysis of disk brakes,
Nonlinear Dynam. {\bf 38} (2004), no.~1-2, 191--206.

\bibitem{15} D. Baleanu\ and\ T. Avkar,
Lagrangians with linear velocities within
Riemann-Liouville fractional derivatives,
Nuovo Cimento {\bf 119} (2004), 73--79.

\bibitem{16} D. Baleanu, S. I. Muslih\ and\ E. M. Rabei,
On fractional Euler-Lagrange and Hamilton equations
and the fractional generalization of total time derivative,
Nonlinear Dynamics, in press.
DOI:10.1007/s11071-007-9296-0;
arXiv:0708.1690v1 [math-ph]

\bibitem{17} D. Baleanu\ and\ J. J. Trujillo,
On exact solutions of a class of fractional Euler-Lagrange equations, Nonlinear Dynamics, in press.
DOI:10.1007/s11071-007-9281-7;
arXiv:0708.1433v1 [math-ph]

\bibitem{cre} J. Cresson,
Fractional embedding of differential operators and Lagrangian
systems, J. Math. Phys. {\bf 48} (2007), no.~3, 033504, 34 pp.

\bibitem{20} R. A. El-Nabulsi,
A fractional approach of nonconservative Lagrangian dynamics,
Fizika A {\bf 14} (2005), no.~4, 289--298.

\bibitem{21} R. A. El-Nabulsi, A fractional action-like variational
approach of some classical, quantum and geometrical dynamics,
Int. J. Appl. Math. {\bf 17} (2005), no.~3, 299--317.

\bibitem{22} R. A. El-Nabulsi, Some geometrical aspects
of fractional nonconservative autonomous Lagrangian mechanics,
Int. J. Appl. Math. Stat. {\bf 5} (2006), no.~S06, 50--64.

\bibitem{27} R. A. El-Nabulsi, Fractional path integral
and exotic vacuum for the free spinor field theory
with Grassmann anticommuting variables,
EJTP, Electron. J. Theor. Phys. {\bf 4} (2007), no.~15, 157--164.

\bibitem{29} R. A. El-Nabulsi,
Some fractional geometrical aspects of weak  field
approximation and Schwarzschild spacetime,
Rom. Journ. Phys. {\bf 52} (2007), no.~5-7, 705--715.

\bibitem{31} R. A. El-Nabulsi, Cosmology with a fractional action principle,
Rom. Rep. Phys. {\bf 59} (2007), no.~3, 759--765.

\bibitem{33} R. A. El-Nabulsi, I. A. Dzenite\ and\ D. F. M. Torres,
Fractional action functional in classical and quantum field theory,
Scientific Proceedings of Riga Technical University,
Series---Computer Science, Boundary Field Problems, and Computer Simulation,
48th thematic issue, 2006, pp.~189--197.

\bibitem{24} R. A. El-Nabulsi\ and\ D. F. M. Torres,
Necessary optimality conditions for fractional
action-like integrals of variational calculus with
Riemann-Liouville derivatives of order
$\left(\alpha,\beta\right)$,
Math. Methods Appl. Sci. {\bf 30} (2007), no.~15, 1931--1939.

\bibitem{23} G. S. F. Frederico\ and\ D. F. M. Torres,
Constants of motion for fractional action-like variational problems,
Int. J. Appl. Math. {\bf 19} (2006), no.~1, 97--104.

\bibitem{23a}  G. S. F. Frederico\ and\ D. F. M. Torres,
Non-conservative Noether's theorem for fractional action-like
variational problems with intrinsic and observer times,
Int. J. Ecol. Econ. Stat. {\bf 9} (2007),
no.~F07, 74--82.

\bibitem{23b} G. S. F. Frederico\ and\ D. F. M. Torres,
A formulation of Noether's theorem
for fractional problems of the calculus of variations,
J. Math. Anal. Appl. {\bf 334} (2007), no.~2, 834--846.

\bibitem{23c} G. S. F. Frederico\ and\ D. F. M. Torres,
Fractional conservation laws in optimal control theory,
Nonlinear Dynamics, in press.
DOI:10.1007/s11071-007-9309-z;
arXiv:0711.0609v1 [math.OC]

\bibitem{23d} G. S. F. Frederico\ and\ D. F. M. Torres,
Necessary optimality conditions for fractional action-like
problems with intrinsic and observer times, WSEAS Transactions on
Mathematics. In Special Issue: Nonclassical Lagrangian Dynamics
and Potential Maps (Guest Editor: C. Udriste),
{\bf 7} (2008), no.~1, 16--22.
arXiv:0712.0152v1 [math.OC]

\bibitem{7} R. Gorenflo\ and\ F. Mainardi,
Fractional calculus: integral and differential equations
of fractional order, in {\it Fractals and fractional calculus
in continuum mechanics (Udine, 1996)}, 223--276, Springer, Vienna.

\bibitem{9} A. A. Kilbas, H. M. Srivastava\ and\ J. J. Trujillo,
{\it Theory and applications of fractional differential equations}, Elsevier, Amsterdam, 2006.

\bibitem{11} M. Klimek, Fractional sequential
mechanics---models with symmetric fractional derivative,
Czechoslovak J. Phys. {\bf 51} (2001), no.~12, 1348--1354.

\bibitem{12} M. Klimek, Lagrangian fractional
mechanics---a noncommutative approach,
Czechoslovak J. Phys. {\bf 55} (2005), no.~11, 1447--1453.

\bibitem{5} R. L. Magin, Fractional calculus in bioengineering,
Part~1, Critic. Rev. in Biomed. Eng. {\bf 32} (2004), no.~1, 110 pp.

\bibitem{6} R. L. Magin, Fractional calculus in bioengineering,
Part~2, Critic. Rev. in Biomed. Eng. {\bf 32} (2004), no.~2, 90 pp.

\bibitem{6b} R. L. Magin, Fractional calculus in bioengineering,
Part~3,  {\bf 32} (2004), no.~3-4, 183 pp.

\bibitem{1} K. B. Oldham\ and\ J. Spanier,
{\it The fractional calculus}, Academic Press
[A subsidiary of Harcourt Brace Jovanovich, Publishers],
New York, 1974.

\bibitem{13} F. Riewe, Nonconservative Lagrangian and Hamiltonian mechanics,
Phys. Rev. E (3) {\bf 53} (1996), no.~2, 1890--1899.

\bibitem{18} F. Riewe, Mechanics with fractional derivatives,
Phys. Rev. E (3) {\bf 55} (1997), no.~3, part B, 3581--3592.

\bibitem{2} S. G. Samko, A. A. Kilbas\ and\ O. I. Marichev,
{\it Fractional integrals and derivatives},
Translated from the 1987 Russian original,
Gordon and Breach, Yverdon, 1993.

\bibitem{3} B. Stankovi\'c,
An equation with left and right fractional derivatives,
Publ. Inst. Math. (Beograd) (N.S.) {\bf 80(94)} (2006), 259--272.

\bibitem{4} V. E. Tarasov,
Fractional variations for dynamical systems:
Hamilton and Lagrange approaches,
J. Phys. A {\bf 39} (2006), no.~26, 8409--8425.

\bibitem{14} V. E. Tarasov, Fractional generalization of gradient
and Hamiltonian systems, J. Phys. A {\bf 38} (2005),
no.~26, 5929--5943.

\bibitem{CD:UdriDu:2005} C. Udriste\ and\ I. Duca,
Periodical solutions of multi-time Hamilton equations, Analele
Universitatii Bucuresti {\bf 55} (2005), no.~1, 179--188.

\bibitem{CD:UdriTe:2007} C. Udriste\ and\ I. Tevy,
Multi-time Euler-Lagrange-Hamilton theory, WSEAS Transactions on
Mathematics {\bf 6} (2007), no.~6, 701--709.

\end{thebibliography}
}
\end{document}